\documentclass[10pt, conference]{IEEEtran}
\IEEEoverridecommandlockouts
\pdfoutput=1
\usepackage{url}
\usepackage[cmex10]{amsmath}
\usepackage{multirow}
\usepackage{epsfig}
\usepackage{mathrsfs}
\usepackage{comment}
\usepackage{array}
\usepackage{amsthm}
\usepackage{blkarray}
\usepackage{enumerate}
\usepackage[lined,ruled,linesnumbered]{algorithm2e}
\usepackage{cite}
\usepackage{amsmath,amssymb,amsfonts}

\usepackage{algorithmic}
\usepackage{graphicx}
\usepackage{lipsum,subcaption}
\usepackage{caption,subcaption}

\usepackage[font=footnotesize,skip=3pt]{caption}

\newtheorem{theorem}{Theorem}

\newtheorem{proposition}[theorem]{Proposition}
\newtheorem{corollary}[theorem]{Corollary}

\newtheorem{definition}[theorem]{Definition}

\def\MS{\mbox{\tiny MS}}
\def\SS{\mbox{\tiny SS}}
\def\PS{\mbox{\tiny PS}}
\def\CS{\mbox{\tiny CS}}

\def\BGP{\mbox{\tiny BGP}}

\def\uniform{\mbox{\tiny uniform}}
\def\nonuniform{\mbox{\tiny nonuniform}}

\newcommand{\prob}{\mbox{Pr}}

\makeatletter
\def\ps@headings{%
\def\@oddhead{\mbox{}\scriptsize\rightmark \hfil \thepage}%
\def\@evenhead{\scriptsize\thepage \hfil \leftmark\mbox{}}%
\def\@oddfoot{}%
\def\@evenfoot{}}
\makeatother
\begin{document}

\title{How Better is Distributed SDN? \\An Analytical Approach\\
\thanks{This research was sponsored by the U.S. Army Research Laboratory and the U.K. Ministry of Defence under Agreement Number W911NF-16-3-0001. The views and conclusions contained in this document are those of the authors and should not be interpreted as representing the official policies, either expressed or implied, of the U.S. Army Research Laboratory, the U.S. Government, the U.K. Ministry of Defence or the U.K. Government. The U.S. and U.K. Governments are authorized to reproduce and distribute reprints for Government purposes notwithstanding any copyright notation hereon.}
}

\author{\IEEEauthorblockN{Ziyao Zhang\IEEEauthorrefmark{1}, Liang 
Ma\IEEEauthorrefmark{2}, Kin K. Leung\IEEEauthorrefmark{1}, Franck Le\IEEEauthorrefmark{2}, Sastry Kompella\IEEEauthorrefmark{3}, and Leandros Tassiulas\IEEEauthorrefmark{4}}
\IEEEauthorblockA{\IEEEauthorrefmark{1}Imperial College London, London, U.K. Email: \{ziyao.zhang15, kin.leung\}@imperial.ac.uk\\
\IEEEauthorrefmark{2}IBM T. J. Watson Research Center, Yorktown Heights, NY, U.S. Email: \{maliang, fle\}@us.ibm.com\\
\IEEEauthorrefmark{3}U.S. Naval Research Laboratory, Washington, DC, U.S. Email: sastry.kompella@nrl.navy.mil\\
\IEEEauthorrefmark{4}Yale University, New Haven, CT, U.S. Email: leandros.tassiulas@yale.edu
}}

\maketitle

\begin{abstract}
Distributed software-defined networks (SDN), consisting of multiple inter-connected network domains, each managed by one SDN controller, is an emerging networking architecture that offers balanced centralized control and distributed operations. Under such networking paradigm, most existing works focus on designing sophisticated controller-synchronization strategies to improve  joint controller-decision-making for inter-domain routing. However, there is still a lack of fundamental understanding of how the performance of distributed SDN is related to network attributes, thus impossible to justify the necessity of complicated strategies. In this regard, we analyze and quantify the performance enhancement of distributed SDN architectures, influenced by intra-/inter-domain synchronization levels and network structural properties. Based on a generic weighted network model, we establish analytical methods for performance estimation under four synchronization scenarios with increasing synchronization cost. Moreover, two of these synchronization scenarios correspond to extreme cases, i.e., minimum/maximum synchronization, which are, therefore, capable of bounding the performance of distributed SDN with any given synchronization levels. Our theoretical results reveal how network performance is related to synchronization levels and inter-domain connections, the accuracy of which are confirmed by simulations based on both real and synthetic networks. To the best of our knowledge, this is the first work quantifying the performance of distributed SDN analytically, which provides fundamental guidance for future SDN protocol designs and performance estimation. 
\end{abstract}
\section{Introduction}
\label{sec:intro}

Software-Defined Networking (SDN) \cite{McKeown}\cite{kreutz2015software}\cite{nunes2014survey}\cite{ethane}, an emerging  networking architecture, significantly improves the network performance due to its programmable network management, easy reconfiguration, and on-demand resource allocation, which has therefore attracted considerable research interests.
One key attribute that differentiates SDN from classical networks is the separation of the SDN's data and control plane. Specifically, in SDN, all control functionalities are implemented and abstracted on the control plane for operational decision making, e.g., flow construction and resource allocation, while the data plane only passively executes the instructions received from the control plane. For a typical SDN architecture, all network decisions are made in the control plane by a control entity, called \emph{SDN controller}, in a centralized manner. Since the centralized SDN controller has the full knowledge of network status, it is able to make the global optimal decision. Yet, such centralized control suffers from major scalability issues. In particular, as a network grows, the number of flow requests and operational constraints are likely to increase exponentially. Such high computation and communication requirements may impose substantial burden on the SDN controller, potentially resulting in significant performance degradation (e.g., delays) or even network failures. 

In this regard, distributed SDN is proposed \cite{gupta2015sdx,kotronis2012outsourcing,kotronis2016stitching,petropoulos2016software,chen2015mlv,thai2013decoupling} to balance the centralized and distributed control. Specifically, a distributed SDN network is composed of a set of subnetworks, referred to as \emph{domains}, each managed by an independent SDN controller.
Moreover, each domain contains several gateways connecting to some other domains; such inter-connected domains then form the distributed SDN architecture. In the distributed SDN, if the controllers do not communicate with each other regarding the network status of their own domains, then the distributed SDN is reduced to the classical multi-AS (Autonomous Systems) network, where the network flows are managed by IGP and BGP protocols. Nevertheless, in the distributed SDN architecture, controllers are expected to exchange information via proactively probing or passive listening. Such additional status information at each controller, called the \emph{synchronized information}, can assist in enhancing decision making for inter-domain tasks. As a special case, when each controller knows the network status in all other domains, i.e., complete synchronization, then all controllers can jointly act as a logically centralized controller, which effectively is the same as the centralized SDN structure as all network decisions are globally optimal. These observations imply that the network performance, e.g., the constructed inter-domain path length, relies heavily on the inter-controller synchronization level. Since complete synchronization among controllers will incur high synchronization costs especially in large networks, practical distributed SDN networks are likely to be able to afford only partial inter-domain synchronization. 

Under partial synchronization, most existing works focus on promoting the inter-domain synchronization so that the final decision making approaches optimality. For instance, information sharing algorithms are proposed in \cite{kotronis2016stitching,petropoulos2016software} for negotiating common traffic policies among various domains. Similarly, efficient frameworks are designed in \cite{chen2015mlv,thai2013decoupling}, aiming to facilitate inter-domain routing selection via fine-grained network status exchanges. 
However, one fundamental question regarding the distributed SDN architecture has generally been ignored: \emph{How does the network performance in distributed SDN relate to network synchronization levels and structural properties?} It is possible that under certain network conditions, e.g., the number of gateways and their connections to external domains, the benefit of increasing the synchronization level is only marginal. Without such fundamental understanding, it is impossible to justify why a complicated mechanism for information sharing or flow construction is necessary in distributed SDN. 
We, therefore, investigate this unsolved yet critical problem in the distributed SDN paradigm, with the goal to quantify its performance under any given network conditions. \looseness=-1

In this paper, we propose a network model to capture intra-/inter-domain connections and non-uniform edge weights in distributed SDN. Such network model is generic in that it only requires degree/weight distribution and the number of gateways in each domain as the input parameters, i.e., it is independent of any specific graph models. 
Based on this network model, we then derive analytical expressions of the network performance, focusing on characterizing the average length (see further discussion in section~\ref{sec:problem_statement}) of the constructed paths with respect to (w.r.t.) random flow requests. Such performance metric is investigated under four canonical synchronization levels, ranging from the minimum to the complete (maximum) synchronization that experience increasing synchronization costs. If a given synchronization scenario cannot be described by any of these four cases, then its performance can always be bounded by our analytical results corresponding to the two extreme cases (i.e., maximum/minimum synchronization). Analytical results reveal that the performance metric is a logarithmic function of the network structural parameters even under the minimum synchronization level. Moreover, the performance gain declines with the increasing synchronization level and the number of gateways. To validate the accuracy of the derived analytical expressions, they are compared against evaluation results using both real and synthetic networks. \looseness=-1
\subsection{Related Work}

The flexibility and scalability of the distributed SDN architecture have stimulated many research efforts in this area. 
In particular, the feasibility of deploying SDN-based mechanisms incrementally to current BGP-glued Internet is considered in \cite{kotronis2012outsourcing}, where routing control planes of multiple domains are outsourced to form centralized control planes for optimizing routing decisions. Similarly, \cite{poularakis2017one} explores the problem of SDN upgrade in ISP (Internet Service Provider) networks
under the constraint of migration costs. 
In addition, protocols and systems, such as  HyperFlow\cite{tootoonchian2010hyperflow} and ONOS\cite{berde2014onos} are proposed to realize logically centralized but physically distributed SDN architecture. Devoflow\cite{curtis2011devoflow} and Kandoo\cite{hassas2012kandoo} are designed to reduce the overheads introduced by the interaction between control and data planes. Moreover, DIFANE\cite{yu2010scalable} and Fibbing \cite{vissicchio2015central} are conceived for limiting the level of centralization and addressing robustness issues, respectively.  From industry research community, Google's B4~\cite{jain2013b4} and Espresso~\cite{yap2017taking}, and Facebook's Edge Fabric~\cite{yap2017taking} are developed to address routing-related challenges in the Internet using SDN-based techniques.
However, most of these works are experiment-based without providing any rigorous mathematical analysis or theoretical guarantees, which therefore motivate us to investigate distributed SDN from the fundamental analytical perspective. \looseness=-1

Since all theoretical results in this paper are obtained based on a weighted graph model, our work is also related to the area of graphical analysis of complex networks. However, most works in these areas are performed on certain graph models and constrained to specific graph properties, e.g., clustering \cite{szabo2003structural}, small-world effect \cite{watts1998collective}, community structure \cite{girvan2002community}, network motif \cite{milo2002network}, scale-free \cite{barabasi2009scale}, etc. In contrast, our network model is substantially generalized with the most relaxed input parameters, i.e., degree and edge distributions. On the other hand, models in \cite{erdos1960evolution,watts1998collective,barabasi1999emergence,newman2001random} are purely randomized, which cannot differentiate intra- and inter-domain links in the context of distributed SDN. Finally, our work is most related to \cite{aleta2017multilayer} and \cite{guo2016levy} as they also consider a layered-network model.
However, the objectives in these papers are the analysis of transport networks and navigation strategies, which are substantially different from our problem. In this paper, we overcome these drawbacks to establish analytical results upon a generic network model that captures all key parameters and synchronization levels in distributed SDN.\looseness=-1

\subsection{Summary of Contributions}
Our main contributions are five-fold.

\emph{1)} We propose a generic two-layer network model capturing intra- and inter-domain connections, and edge weights;

\emph{2)} On top of the network model in \emph{1)}, we study the average length of constructed paths (APL) as the performance metric, and develop the analytical expression (a logarithmic function) of the APL under the minimum synchronization level. 
This result serves as an upper bound of the performance metric; \looseness=-1

\emph{3)} We derive a mathematical expression of the APL when the synchronization level is between the minimum and maximal synchronizations. These expressions give fine-grained quantification on how the performance metric relates to the incremental synchronization changes; 

\emph{4)} We establish an analytical expression of the APL for the maximum synchronization scenario, where all domains are synchronized with each other, i.e., complete synchronization. The theoretical result under such synchronization level provides a lower bound of the performance metric;

\emph{5)} All of above theoretical results are evaluated using real and synthetic networks, both of which confirm their high accuracy as well as their capability in providing new insights into performance changes over various network conditions.

In this paper, 
we do not intend to design improved inter-domain routing mechanisms, and thus only basic and typical routing strategies are employed for theoretical analysis under each of the synchronization scenarios. To the best of our knowledge, this is the \emph{first work} that studies distributed SDN from the analytical perspective. The significance of these results is that they lay a strong theoretical foundation for the research community in distributed SDN.

The rest of the paper is organized as follows. Section~\ref{sec:problemFormulation} formulates the problem. Sections~\ref{sec:APL_MS}--\ref{sec:APL_CS} present analytical results for four different synchronization scenarios, respectively. Evaluations of the derived analytical expressions are conducted in Section~\ref{sec:evaluations}. Finally, Section~\ref{sec:Conclusions} concludes the paper.

\section{Problem Formulation}
\label{sec:problemFormulation} 
\subsection{Network Model}
\label{sec:network_model}
We formulate the distributed SDN network as an undirected graph according to a two-layer network model (Fig.~\ref{fig:2-tier}), where the top-layer abstracts the inter-domain connections, and under such constraints, the bottom-layer characterizes physical connections among all network elements. Specifically, the top-layer is a graph consisting of $m$ vertices, where each vertex represents a domain in the distributed SDN. These $m$ vertices are connected via undirected links according to a given \emph{inter-domain degree distribution}, which refers to the distribution of the number of neighboring domains of an arbitrary domain. The top-layer graph, denoted by $\mathcal{G}_d=(V_d,E_d)$ ($V_d$/$E_d$: set of vertices/edges in $\mathcal{G}_d$, $|V_d|=m$), is called \emph{domain-wise topology} in the sequel. The existence of an edge in $E_d$ connecting two vertices $v_1,v_2\in V_d$ in the domain-wise topology implies that the two network domains corresponding to $v_1$ and $v_2$ are connected. Based on this domain-wise topology, we next construct the physical network in the bottom-layer. In particular, each of the $m$ domains in $\mathcal{G}_d$ corresponds to an undirected graph with $n$ nodes in the bottom-layer; these $n$ nodes are connected following a given \emph{intra-domain degree distribution}, which is the distribution of the number of neighboring nodes of an arbitrary node within the same domain.\footnote{In one domain, some nodes may have connections to other domains; such external connections are not considered in the concept of intra-domain degree.} We also assume that such intra-domain degrees across all domains are independently and identically distributed (i.i.d.). The graph of each domain is referred to as \emph{intra-domain topology}. Then for each $e\in E_d$ with end-points corresponding to  domains $\mathcal{A}_i$ and $\mathcal{A}_j$, we (i) randomly select two nodes $w_1$ from $\mathcal{A}_i$ and $w_2$ from $\mathcal{A}_j$ and connect these two nodes if link $w_1w_2$ does not exist, and (ii) repeat such link construction process between $\mathcal{A}_i$ and $\mathcal{A}_j$ $\beta$ times. By this link construction process, the bottom-layer network topology $\mathcal{G}=(V,E)$ is therefore formed ($V/E$: set of nodes/links in $\mathcal{G}$, $|V|=mn$); see Fig.~\ref{fig:2-tier} for illustrations. In each domain, nodes having connections to other domains are called \emph{gateways}. Note that the above process indicates that the $i$-th selected link may overlap with existing links (i.e., the same end-points); therefore, parameter $\beta$ represents the maximum number of links between any two domains. Hence, if two domains, each with $n$ nodes, are connected in the domain-wise topology, then the expected number of links connecting these two domains is $n^2\big(1-(1-\frac{1}{n^2})^{\beta}\big)$. Without loss of generality, we assume that all inter/intra-domain topologies are connected graphs. \looseness=-1

\begin{figure}[tb]
\centering
\includegraphics[width=3.3in]{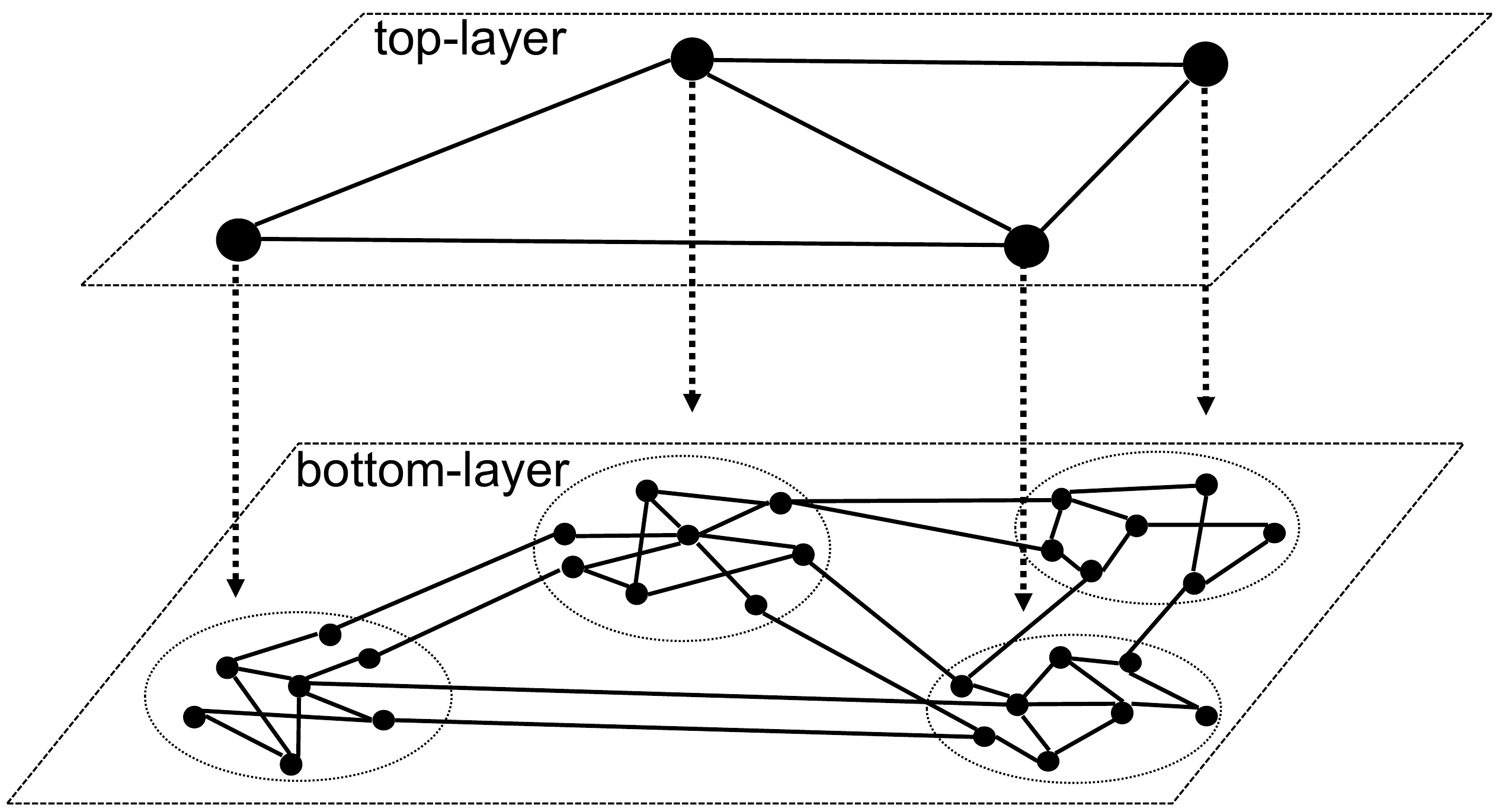}
\caption{Two-layer network model: Top-layer abstracts the domain-wise topology; bottom-layer determines all physical connections in the network.}
\label{fig:2-tier}
\end{figure}

In addition,
we also associate weights to links to capture the corresponding cost, e.g., computation, storage and/or communication cost, used for flow constructions. Specifically, 
we assume intra-domain link weights across all domains are non-negative and i.i.d.. In real distributed SDN environment, unlike the potential wireless links within a domain, inter-domain gateway-to-gateway links are likely to be wired with high bandwidth, thus more stable. In this regard, we characterize all inter-domain link weights by a non-negative constant $C$. Furthermore, without loss of generality, we assume $C=1$; all theoretical results in this paper can be trivially extended to other values of $C$.

\emph{Discussions:} 
Our two-layer network model is generic in that the inputs can be any degree and edge distributions; such distributions can be empirical or extracted from real networks of interest.
Moreover, we do not require inter-domain degree and intra-domain degree to follow the same distribution. \looseness=-1
\subsection{SDN Data and Control Plane}
Thus far, we have only discussed the graphical properties of the distributed SDN networks. One critical aspect of SDN that differentiates it from other networks is the separation of the data and control planes, which are formulated as follows. 

\subsubsection{Data Plane} We exploit graph $\mathcal{G}$ generated by the two-layer network model in Section~\ref{sec:network_model} to represent the data plane of the distributed SDN. Specifically, a node/link exists in $\mathcal{G}$ if and only if it can be used for data transmission in the network.

\subsubsection {Control Plane}
We assume that in this two-layer network model, each domain contains one SDN controller that carries out control operations and facilitates information sharing. Each SDN controller can be one or a collection of intra-domain nodes that are equipped with the controlling functionality (i.e., in-band control \cite{Schiff2016}) or external controlling entities operating on top of a network domain (i.e., out-of-band control \cite{feamster2014}). SDN controllers together with all inter/intra-domain controlling channels form a control plane.

In such network structure, to construct a path between a  pair of source and destination nodes (only \emph{unicast} is considered in this paper), the corresponding routing path is determined by the controllers in the source, destination, and all intermediate domains collectively. However, the performance of the constructed paths may vary, depending on the network status information at each involved controller.

\emph{Remark:} We do not specify the locations of SDN controllers, as they do not affect our theoretical analysis. For the same reason, we do not explicitly visualize the control plane in Fig.~\ref{fig:2-tier} or the rest of this paper. 
\subsection{Synchronization Among SDN Controllers} 
\label{sec:synchScenarios}
As discussed, synchronization levels among SDN controllers directly affects the quality of the constructed paths. We now formally define \emph{synchronization} among SDN controllers.
\begin{definition}
\label{def:synch}
Domain $\mathcal{A}_i$ is \emph{synchronized} with domain $\mathcal{A}_j$ if and only if the SDN controller in $\mathcal{A}_i$ knows the shortest distance (measured by the accumulated weight of shortest path) between any two gateways in $\mathcal{A}_j$.
\end{definition}
By Definition~\ref{def:synch}, clearly there exist exponentially many synchronization levels, i.e., which domains are synchronized with which other domains. However, in real networks, it is usually the case that synchronization difficulty is high when two SDN controllers are far apart. 
In this paper, we study the following synchronization scenarios, sorted by their corresponding synchronization costs. A synchronization scenario corresponds to a specific set of synchronization levels in all domains.  \looseness=-1

\begin{enumerate}[(a)]
\item \emph{{Minimum Synchronization (MS)}}: MS corresponds to the minimal synchronization level. Under MS, no domains synchronize with any other domains. As a result, each controller only knows its own intra-domain topology and the domain-wise topology, but has no knowledge of its intra-domain link weights. This scenario captures IGP routing protocols that do not take into account any link weights but select routes purely based on the hop count  (e.g., Routing Information Protocol (RIPv2)). 
Note that MS corresponds to the minimum network knowledge that is always available,\footnote{This is a valid assumption in existing multi-domain networks, where BGP-like protocols are being used. Specifically, under BGP, intra-domain topology is obtained via techniques such as BGP route reflection \cite{bates2006bgp}, while the domain-wise topology is obtained by external BGP \cite{chandra1996bgp}.} including in scenarios in (b--d) ; 
\item \emph{{Self-domain Synchronization (SS)}}: In addition to the available knowledge provided by MS,  each controller with SS knows nothing more except for the intra-domain link weights (\emph{not} the distribution) in its own domain. With this additional information, one controller can find the optimal intra-domain path for any intra-domain flow request;
\item \emph{{Partial Synchronization (PS)}}: PS refers to any synchronization level that is between SS and the following complete synchronization (CS); 
\item \emph{{Complete Synchronization (CS)}}: under CS, every pair of domains $\mathcal{A}_i$ and $\mathcal{A}_j$ synchronize with each other.
As such, there is effectively one logically centralized controller, which can make globally optimal decisions. Among all these synchronization scenarios, CS experiences the highest synchronization cost.
\end{enumerate}
             
    
    
    
\subsection{Routing Mechanisms}
\label{sec:routing_mechanisms}
We describe a path construction mechanism for each of the above synchronization scenarios (see Sections~\ref{sec:APL_MS}--\ref{sec:APL_CS} for details). The aim of these path construction mechanisms is to minimize the total length of the constructed path between two given nodes. 
Though these routing mechanisms are different, the common rule that governs them is that given a particular synchronization level, each controller makes its own decision as to which nodes and (or) links in its own domain should be selected to construct the source-destination path upon request.
Then the selected path segments in all participating domains concatenate into a cross-domain, end-to-end path. 
\subsection{Problem Statement and Objective}
\label{sec:problem_statement}
Given the distributed SDN network model in Section~\ref{sec:network_model}, our goal is to study the performance of the paths constructed by the routing mechanisms for various synchronization scenarios. In real networks, the performance of routing can be measured by many metrics, such as delay, congestion level, and the number of flows that can be served at the same time, depending on the goal of network management. In order to make our analytical work sufficiently generalized to capture the performance concern that is fundamental to most network management tasks, we exploit the \emph{Average Path Length (APL)}, measured by the average end-to-end accumulated weight of the paths constructed between two arbitrary nodes within different domains in an arbitrary network realization following the two-layer network model, as the performance metric. APL is a natural generalized performance metric, as link weights can always be manipulated to reflect different routing objectives\footnote{APL is of special significance to BGP, as AS-PATH and NEXT-HOP attributes in BGP are both related to APL.}.
Formally, our research objective is:

\textbf{\emph{Objective:}} Suppose (i) each network realization following the two-layer network model exists with the same probability, and (ii) the source-destination node pair belonging to two different domains in a given network realization also exist with the same probability. Our goal is to derive the mathematical expression of 
APL for each of the four synchronization scenarios, namely MS, SS, CS, and PS, in Section~\ref{sec:synchScenarios}. 

Note that we are only interested in studying the cross-domain routing here. For intra-domain routing, the corresponding controller can easily find the optimal paths without relying on inter-controller synchronizations.\looseness=-1

\emph{Remark:} It is important to notice that our two-layer network model is a random graph model, i.e., there exist multiple network realizations satisfying the same set of input parameters. Therefore, APL is an expected value over not only random source/destination node pairs but also random network realizations. All our theoretical results on APL are based on the given network parameters (e.g., degree and weight distributions) rather than a specific network realization.

Main notations used in this paper are summarized in Table~\ref{tb:main_notations}.

\begin{table}[tb]
\renewcommand{\arraystretch}{1.3}
\caption{Main Notations.} 
\label{tb:main_notations}
\centering
\begin{tabular}{r|m{6.5cm}}
  \hline
  \textbf{Symbol} & \textbf{Meaning} \\
  \hline
  $n$ &  number of nodes in a domain\\
  \hline
  $m$ &  number of domains in the network \\
  \hline
   $\beta$ & maximum number of edges connecting two domains\\
 \hline
  $l$ & average shortest path length between a non-gateway and the closest gateway in a domain\\
  \hline
  $\Delta$ & average shortest domain-wise path length between two arbitrary domains \\
 \hline
 $\mathcal{D}$ & random variable representing the shortest distance between two random nodes in a domain\\
 \hline
  $D_k^{(\beta)}$ & random variable representing the APL between two arbitrary nodes in the end-domains of a bus network of length $k$\\
  \hline
  $M_k^{(\beta)}$ & random variable representing the APL between an arbitrary node in an end-domain and the closest gateway (connecting to external domains) in the other end-domain of a bus network of length $k$\\
  \hline
  $L_k(\beta)$ & expectation of random variable $D_k^{(\beta)}$\\
  \hline
\end{tabular}
\end{table}

    
\section{Average Path Length under Minimum Synchronization}
\label{sec:APL_MS}

In this section, we study the APL under MS, for which we  describe the corresponding routing mechanism and then present its performance analysis. For ease of presentation, we first introduce the following definitions and notations. 
\begin{definition}
\begin{enumerate}
\item In the domain-wise topology $\mathcal{G}_d$, the vertex corresponding to domain $\mathcal{A}$ in $\mathcal{G}$ is denoted by $\vartheta(\mathcal{A})$;
\item Given a pair of source and destination nodes $v_1$ and $v_2$ with $v_1\in \mathcal{A}_1$, $v_2\in \mathcal{A}_2$, and $\mathcal{A}_1\neq \mathcal{A}_2$, the domain-wise path w.r.t. $v_1$ and $v_2$ is a path in $\mathcal{G}_d$ starting at vertex $\vartheta(\mathcal{A}_1)$ and terminating at vertex $\vartheta(\mathcal{A}_2)$;
\item The shortest path between node $w$ and set $S$ is the shortest path in set $\{\mathcal{P}_s$: shortest path between node $w$ and node $s$, $s\in S\}$.
\end{enumerate}
\end{definition}
\subsection{Routing Mechanism under MS (RMMS)}
\label{sec:Routing_MS} 
A path between two nodes $v_1$ and $v_2$ can be constructed in the following way, called \emph{Routing Mechanism under MS (RMMS)}, which is similar to the BGP for inter-domain routing and the IGP for intra-domain routing\footnote{Note that BGP routing policies typically reflects the commercial agreements among the domains. For generality purpose, we do not consider these factors in our routing mechanisms.}. 

\begin{enumerate}[(i)]
\item Select the shortest domain-wise path w.r.t. $v_1$ and $v_2$, which goes through vertices $\vartheta(\mathcal{A}_1),\vartheta(\mathcal{A}_2),\ldots,\vartheta(\mathcal{A}_q)$ in $\mathcal{G}_d$ ($v_1\in\mathcal{A}_1,v_2\in\mathcal{A}_q$), with ties (if any) broken arbitrarily. In other words, no domain-wise path from $\vartheta(\mathcal{A}_1)$ to $\vartheta(\mathcal{A}_q)$ traverses less than $q$ domains.
\item In domain $\mathcal{A}_i$, let $w$ be (i) the source node $v_1$ if $i=1$, or (ii) the starting point (also known as ingress node) selected by the controller in $\mathcal{A}_{i-1}$ if $i\geq 2$. In addition, let (i) set $S=\{v_2\}$ if $i=q$, or (ii) $S=\{g: g \text{ is a gateway that has a link connecting to } \mathcal{A}_{i+1}\}$ if $i\leq q-1$. Then the controller in $\mathcal{A}_i$ selects the shortest path\footnote{This is similar to several IGPs, such as Open Shortest Path First (OSPF), RIP, RIPv2, and Intermediate System to Intermediate System (IS-IS) \cite{vetriselvan2014survey}.} $\mathcal{P}_i$, in terms of the number of hops (link weights are not available, thus not considered) with ties broken randomly, from node $w$ to set $S$ using only nodes and links in $\mathcal{A}_i$. Let $s'\in S$ denote the end-point of $\mathcal{P}_i$. If $i\leq q-1$, then the controller in $\mathcal{A}_i$ further appends a random node $u\in\{w:w\in \mathcal{A}_{i+1}, s'w\in \mathcal{G}\}$ to path $\mathcal{P}_i$, thus forming path $\mathcal{P}'_i$ ($u$ is then the ingress node in $\mathcal{A}_{i+1}$); otherwise $\mathcal{P}'_i=\mathcal{P}_i$.
\item The final end-to-end path $\mathcal{P}$ connecting $v_1$ and $v_2$ is 
\begin{equation}
\label{eq:pathConcatenation}
\mathcal{P}=\mathcal{P}'_1+\mathcal{P}'_2+\ldots +\mathcal{P}'_q; 
\vspace{-.5em}
\end{equation}
see Fig.~\ref{fig:MS_example} for the example.
\end{enumerate}

\begin{figure}[tb]
\centering
\includegraphics[width=3.5in]{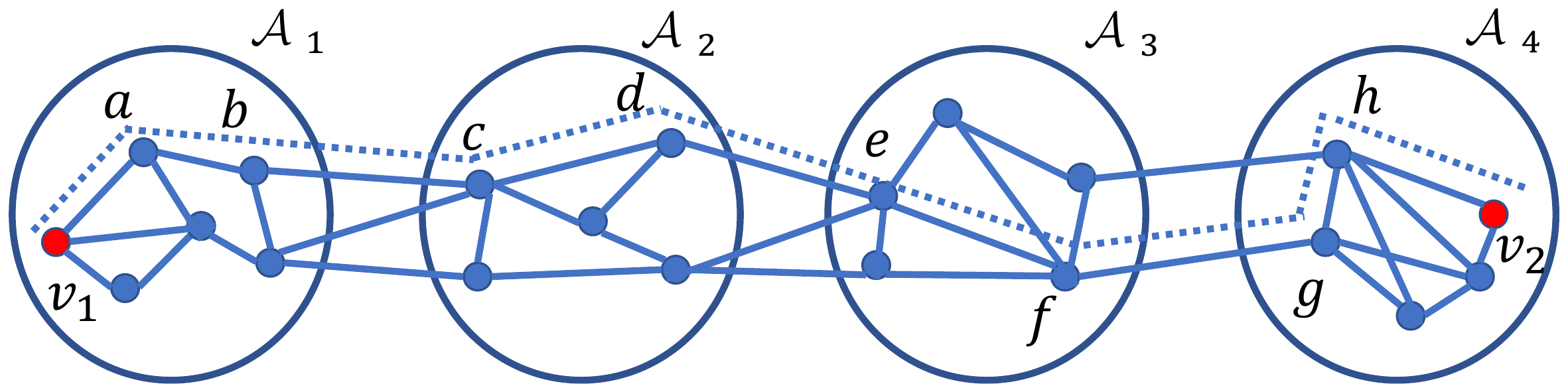}
\caption{Path construction under RMMS w.r.t. $v_1$ and $v_2$, whose shortest domain-wise path traverses $\mathcal{A}_1$, $\mathcal{A}_2$, $\mathcal{A}_3$, and $\mathcal{A}_4$. Then $\mathcal{P}_1 = v_1ab$, $\mathcal{P}'_1 = v_1abc$, $\mathcal{P}_2 = cd$, $\mathcal{P}'_2 = cde$, $\mathcal{P}_3=ef$, $\mathcal{P}'_3=efg$, $\mathcal{P}_4 = \mathcal{P}'_4 =ghv_2$. The constructed path is $\mathcal{P} = \mathcal{P}'_1 +\mathcal{P}'_2 + \mathcal{P}'_3 +\mathcal{P}'_4$, denoted by dotted line segments.\looseness=-1}
\label{fig:MS_example}
\end{figure}

\subsection{APL under Minimum Synchronization (MS)} 
\label{subsec:APL_MS}
With RMMS, we are now ready to analyze APL for MS. The basic idea is that we first compute the average domain-wise path length w.r.t. two arbitrary source/destination nodes. Then in a domain-wise path with such average length, we calculate the average number of hops in each traversed domain, and add them together to get the final estimation of APL for MS.
To this end, we first present the results in the existing work \cite{newman2001random} that can assist our mathematical analysis. 

\begin{proposition}\cite{newman2001random}
\label{prop:existingWork}
In an undirected graph $\mathcal{H}$ with $n_0$ vertices and the vertex degree satisfying a given distribution, let $x_i$ be the average number of vertices that are $i$-hop away from a random vertex in $\mathcal{H}$. Suppose all edge weights are $1$, and $x_2\gg x_1$. Then
\begin{enumerate}
\item 
\begin{equation}
\label{eq:xi}
x_{i} = ({x_{2}}/{x_{1}})^{i-1}x_{1};
\vspace{-.5em}
\end{equation}
\item APL in $\mathcal{H}$ is 
\begin{equation}
\label{eq:APLinH}
\frac{\log(n_0/x_{1})}{\log(x_{2}/x_{1})} + 1.
\vspace{-.5em}
\end{equation}
\end{enumerate}
\end{proposition}

In our two-layer model, the top-layer graph $\mathcal{G}_d$ (domain-wise topology with $m$ vertices) itself is a random graph following a given domain-wise degree distribution. Therefore, similar to \cite{newman2001random}, let $z'_i$ denote the average number of vertices that are $i$-hop away from a random vertex in $\mathcal{G}_d$. For two arbitrary nodes $v_1$ and $v_2$ with $v_1\in \mathcal{A}_1$, $v_2\in \mathcal{A}_q$, and $\mathcal{A}_1 \neq \mathcal{A}_q$, let $\Delta$ denote the average number of hops of the shortest domain-wise path from vertex $\vartheta(\mathcal{A}_1)$ to vertex $\vartheta(\mathcal{A}_q)$ in $\mathcal{G}_d$.
Then according to (\ref{eq:APLinH}), we have

\begin{equation}
\label{eq:Delta}
\Delta =\frac{\log(m/z_{1}^{\prime})}{\log(z_{2}^{\prime}/z_{1}^{\prime})} + 1,
\end{equation}
assuming $z'_2\gg z'_1$.

With (\ref{eq:Delta}), we know that the average value of $q$ in (\ref{eq:pathConcatenation}) is $\Delta+1$. If we further know the average length of $\mathcal{P}'_i$ associated with the traversed domain $\mathcal{A}_i$, then we can estimate the APL of $\mathcal{P}$. To this end, let $|\mathcal{P}|$ denote the \emph{number of hops} on path $\mathcal{P}$. Then according to (\ref{eq:pathConcatenation}), $|\mathcal{P}|=|\mathcal{P}'_1|+|\mathcal{P}'_2|+\ldots +|\mathcal{P}'_{\Delta+1}|$, where $q$ is replaced by its expection $\Delta+1$, and $|\mathcal{P}'_i|$ is a random variable. The expectation of $|\mathcal{P}|$ is
\vspace{-.5em}
\begin{equation}
\label{eq:pathHopExpLong}
\begin{split}
\mathbb{E}[|\mathcal{P}|]&=\mathbb{E}[|\mathcal{P}'_1|+|\mathcal{P}'_2|+\ldots +|\mathcal{P}'_{\Delta+1}|]\\
&= \mathbb{E}[|\mathcal{P}'_1|]+\mathbb{E}[|\mathcal{P}'_2|]+\ldots +\mathbb{E}[|\mathcal{P}'_{\Delta+1}|].
\end{split}
\vspace{-2em}
\end{equation}
According to the path construction procedure for MS, $\mathbb{E}[|\mathcal{P}'_1|]=\mathbb{E}[|\mathcal{P}'_2|]=\ldots =\mathbb{E}[|\mathcal{P}'_{\Delta}|]$ for two reasons. First, all domains have the same statistical properties. Second, in each domain $\mathcal{A}_i$ ($i\leq \Delta$), the routing mechanism selects a gateway (from a candidate set) that is closest to the ingress node. By contrast, in domain $\mathcal{A}_{\Delta+1}$, the routing mechanism only selects the shortest path from the ingress node to a single node $v_2$, i.e., the destination. Thus, (\ref{eq:pathHopExpLong}) can be simplified as
\begin{equation}
\label{eq:pathHopExpShort}
\mathbb{E}[|\mathcal{P}|]= \Delta\cdot\mathbb{E}[|\mathcal{P}'_1|]+\mathbb{E}[|\mathcal{P}'_{\Delta+1}|].
\end{equation}
In a domain $\mathcal{A}$ with $n$ intra-domain nodes, let $z_i$ denote the average number of intra-domain nodes that are $i$-hop ($i\geq 1$) away from an arbitrary node $v$ ($v\in \mathcal{A}$). Then again by (\ref{eq:APLinH}), we have
\begin{equation}
\label{eq:lastDomain}
\mathbb{E}[|\mathcal{P}'_{\Delta+1}|]=\frac{\log(n/z_{1})}{\log(z_{2}/z_{1})} + 1,
\end{equation}
assuming $z_2\gg z_1$. Hence, to compute $\mathbb{E}[|\mathcal{P}|]$ in (\ref{eq:pathHopExpShort}), it suffices to consider only $\mathbb{E}[|\mathcal{P}'_1|]$ associated with domain $\mathcal{A}_1$.

Recall that $\mathcal{P}'_1$ is obtained by intra-domain path $\mathcal{P}_1$ in $\mathcal{A}_1$ and also an inter-domain connection to $\mathcal{A}_2$; hence, $\mathbb{E}[|\mathcal{P}'_1|]=\mathbb{E}[|\mathcal{P}_1|]+1$. We can therefore focus on computing $\mathbb{E}[|\mathcal{P}_1|]$.
In $\mathcal{A}_1$, on average, there are $\gamma = n(1-(1-1/n)^{\beta})$ gateways connecting to $\mathcal{A}_2$. Suppose $\mathcal{A}_1$ contains exactly $\gamma$ gateways, denoted by set $S$. Then regarding path $\mathcal{P}_1$ from the starting point $v_1$ in $\mathcal{A}_1$ to set $S$, there are two cases. First, $v_1\in S$, then $\mathcal{P}_1$ is a degenerate path containing only one node $v_1$, i.e., $|\mathcal{P}_1|=0$. Second, $v_1\notin S$, which complicates the computation of $|\mathcal{P}_1|$. For the second case, 
let $l:=\mathbb{E}[|\mathcal{P}_1|\ |v_1\notin S]$, i.e., the expectation of $|\mathcal{P}_1|$ conditioned on $v_1\notin S$.
Regarding the gateway set $S$, there are up to $\gamma z_i$ non-gateways that are $i$-hop away from the closest gateways. Let $l_{\max}:=\arg\max_i z_i$ s.t. $\gamma + \sum_j z_j\leq n$. According to (\ref{eq:xi}), $z_i$ increases exponentially with $i$. In other words, the majority of non-gateways are $l_{\max}$-hop away from the closest gateways; therefore, we use $l_{\max}$ to approximate $l$. Thus,
$z_l\approx z_{l_{\max}}\approx n-\gamma \approx n+1-\gamma$ when $n$ is large. By solving $z_l = n+1-\gamma$, we obtain
\begin{equation}
\label{eq:l}
l = \frac{\log(\frac{n+1-\gamma}{z_1\gamma})}{\log(z_{2}/z_{1})} +1,
\end{equation}
where $\gamma = n(1-(1-1/n)^{\beta})$.
By close examination of (\ref{eq:l}), we notice that it is also needed to guarantee $l\geq 1$. Hence, (\ref{eq:l}) can be calibrated as follows.
\begin{equation}
\label{eq:finalL}
l=\left\{
                \begin{array}{ll}
                 \frac{\log(\frac{n+1-\gamma}{z_1\gamma})}{\log(z_{2}/z_{1})} +1 & \text{if}\ \gamma \leq \frac{n+1}{z_1+1},\\
                  1  & \text{otherwise}.\\
                \end{array}
              \right.
\end{equation}
It can be verified that when $\gamma=1$, (\ref{eq:finalL}) reduces to (\ref{eq:lastDomain}) as expected. In (\ref{eq:finalL}), it reveals a key threshold $\gamma_0=(n+1)/(z_1+1)$. When $\gamma\leq \gamma_0$, the distance from an arbitrary non-gateway to the closest gateway is relatively large; nevertheless, when $\gamma>\gamma_0$, there are sufficiently many gateways randomly distributed in one domain, causing each non-gateway having a gateway neighbor with high probability. Hence,
\begin{equation}
\label{eq:EP1prime}
\begin{split}
\mathbb{E}[|\mathcal{P}'_1|]&=\mathbb{E}[|\mathcal{P}_1|]+1\\
&=\mathbb{E}[|\mathcal{P}_1|\ |v_1\notin S]\prob(v_1\notin S) \\
&\ \ \ + \mathbb{E}[|\mathcal{P}_1|\ |v_1\in S]\prob(v_1\in S) +1\\
&=\frac{l(n-\gamma)}{n}+1.
\end{split}
\end{equation}
Putting (\ref{eq:Delta}), (\ref{eq:lastDomain}), and (\ref{eq:EP1prime}) into (\ref{eq:pathHopExpShort}), we get $\mathbb{E}[|\mathcal{P}|]$.

Thus far, we have not considered link weights in the network, because this information is not available to the routing mechanism under MS. However, the APL for MS, denoted by $L_{\MS}$, needs to account for the link weights. Recall in our two-layer model, all intra-domain link weights are modeled as a given i.i.d. random variable, denoted by $W$, and all inter-domain edges are of weight $1$. Hence,
\begin{equation}
\label{eq:L_MS}
\begin{split}
L_{\MS}&=\Delta\cdot(\mathbb{E}[|\mathcal{P}_1|]\cdot\mathbb{E}[W]+1)+\mathbb{E}[|\mathcal{P}'_{\Delta+1}|]\cdot\mathbb{E}[W]\\
&=\bigg(\frac{(n-\gamma)l\Delta}{n}+\frac{\log (n/z_1)}{\log (z_2/z_1)}+1\bigg)\mathbb{E}[W]+\Delta.
\end{split}
\end{equation} 
In the multi-domain SDN network, $L_{\MS}$ represents the APL under MS, which is the worst case, i.e., the longest paths. Nevertheless, (\ref{eq:L_MS}) shows that even such worst case can be quantified by a logarithmic function of network parameters. In the following sections, we investigate to what extent APL can be improved when more synchronized information is available.
    
    

\section{Average Path Length under Self-domain Synchronization}
\label{sec:APL_SS}

Similar to MS, no two domains  synchronize under self-domain synchronization (SS). Therefore, the routing mechanism under SS is almost the same as RMMS, except that each controller selects the shortest intra-domain path in terms of the accumulated link weight.

To analyze APL under SS, 
we need to combine the distributions of link weights and the number of hops into a new distance distribution for capturing APL. Here is the sketch of our analytical methodology.

\emph{Sketch of Analytical Methodology:} 
\vspace{-.3em}
\begin{enumerate}[a)]
\item We first compute the distribution of the distance between two random nodes within the same domain, called \emph{intra-domain distance distribution};
\item The expression in (\ref{eq:pathConcatenation}) remains valid. Therefore, we need to determine the APL of $\mathcal{P}'_i$ for $i=1,2,\ldots,q$. Since the domain-wise path is selected the same way as that under MS, again we have that the expected value of $q$ is $\Delta+1$;
\item As all controllers involved in the path construction process follow the same procedure, similar to (\ref{eq:pathHopExpShort}), it suffices to only quantify the APL of $\mathcal{P}'_1$ and $\mathcal{P}'_{\Delta+1}$ using the intra-domain distance distribution derived in a).
\end{enumerate}
\vspace{-.3em}
Based this methodology, we next discuss the details.
\vspace{-.3em}
\subsection{Intra-Domain Distance Distribution}
\vspace{-.3em}
In one domain, consider a path with $\lambda$ links. Let $W_1,W_2,\ldots,W_\lambda$ be i.i.d. random variables of link weights on this path with the probability density functions (pdf) being $f_{W_1}(x)=f_{W_2}(x)=\ldots =f_{W_{\lambda}}(x)$. Define random variable $\mathcal{W}_\lambda:=\sum^{\lambda}_{i=1}W_i$ as the accumulated weight on this path. Then the pdf of $\mathcal{W}_\lambda$ is the convolution of the pdfs of  $W_1,W_2,\ldots,W_\lambda$, i.e., $f_{\mathcal{W}_\lambda}(x)=f_{W_1}(x)*f_{W_2}(x)*\ldots *f_{W_{\lambda}}(x)$. 
By the principle in mixture distribution \cite{mixture}, we still need to know the probability $p_{\mathcal{W}_{\lambda}}$ of the shortest path (in terms of weights) between two random nodes containing $\lambda$ links. By the concept of $z_i$ we defined in the analysis of MS, we know that $p_{\mathcal{W}_{\lambda}}$ is proportional to $z_{\lambda}$; therefore, by normalization, we get $p_{\mathcal{W}_{\lambda}}=z_{\lambda}/n$. Note that when $\lambda=0$, $z_0=1$ and the cumulative distribution function (cdf) of $\mathcal{W}_0$ is a unit step function \cite{unitstepfunction}. Let random variable $\mathcal{D}$ be the shortest distance (in term of overall weights) between two random nodes in one domain, with the pdf being $f_{\mathcal{D}}(x)$, i.e., intra-domain distance distribution. Then by mixture distribution, $f_{\mathcal{D}}(x)$ can be estimated as follows.
\vspace{-.5em}
\begin{equation}
\label{eq:InterDomainDistance}
f_{\mathcal{D}}(x) = \sum_{i=0}^{h_{\max}}p_{\mathcal{W}_i}f_{\mathcal{W}_i}(x)=\sum_{i=0}^{h_{\max}}\frac{z_i}{n}\cdot f_{\mathcal{W}_i}(x),
\vspace{-.5em}
\end{equation}
where $h_{\max}:=\arg\max_i z_i$ s.t. $\sum_{i=0}^{h_{\max}}z_i\leq n$. Hence, the APL between two nodes in one domain is 
\begin{equation}
\label{eq:expD}
\mathbb{E}[\mathcal{D}]=\int_{x=0}^{+\infty}xf_{\mathcal{D}}(x).
\vspace{-1em}
\end{equation}


\vspace{-.3em}
\subsection{Domain-wise Path}
\vspace{-.3em}
Though SS and MS represent different synchronization levels, the corresponding domain-wise paths are exactly the same w.r.t. a pair of source and destination nodes in a given network. Thus, the end-to-end path construction can still be captured by (\ref{eq:pathConcatenation}). By (\ref{eq:Delta}), again, we have $q$ in (\ref{eq:pathConcatenation}) equals $\Delta +1$. Let $L(\mathcal{P})$ be the end-to-end accumulated weight (i.e., length) of path $\mathcal{P}$, which is a random variable. Then the expectation of $L(\mathcal{P})$, i.e., the APL for SS, denoted by $L_{\SS}$ is
\vspace{-.5em}
\begin{equation}
\label{eq:pathLengthExpSS}
\begin{split}
L_{\SS}&=\mathbb{E}[L(\mathcal{P})]\\
&=\mathbb{E}[L(\mathcal{P}'_1)+L(\mathcal{P}'_2)+\ldots +L(\mathcal{P}'_{\Delta+1})]\\
&= \Delta\cdot\mathbb{E}[L(\mathcal{P}'_1)]+\mathbb{E}[L(\mathcal{P}'_{\Delta+1})]\\
&=\Delta\cdot\mathbb{E}[L(\mathcal{P}'_1)]+\mathbb{E}[\mathcal{D}].
\end{split}
\vspace{-1em}
\end{equation}
The reason for the last row in (\ref{eq:pathLengthExpSS}) is that $\mathbb{E}[L(\mathcal{P}'_{\Delta+1})]$ essentially is the APL between two nodes in one domain. 
Thus, it suffices to determine $\mathbb{E}[L(\mathcal{P}'_1)]$ next, i.e., distance from a random starting point to the closest gateway within the domain.\looseness=-1
\vspace{-.4em}
\subsection{Distance from a Node to the Closest Gateway}
\vspace{-.3em}
Let random variable $M^{(\beta)}$ denote the shortest distance from an arbitrary node $w$ to the closest gateway in the candidate gateway set $S$ within a domain with the inter-domain connection parameter $\beta$, where $S$ contains all gateways connecting to the same neighboring domain. 
Recall that in our two-layer model, gateways are randomly selected. Therefore, let $\mathcal{D}_1,\mathcal{D}_2,\ldots, \mathcal{D}_\beta$ be i.i.d. random variables, denoting the shortest distance between two random nodes in a domain with the same pdf as $\mathcal{D}$ in (\ref{eq:InterDomainDistance}). We have 
\vspace{-.5em}
\begin{equation}
\label{eq:minRV}
M^{(\beta)}=\min(\mathcal{D}_1,\mathcal{D}_2,\ldots, \mathcal{D}_\beta),
\vspace{-.5em}
\end{equation} 
and $L(\mathcal{P}'_1)=M^{(\beta)}+1$ as $\mathcal{P}'_1$ terminates at domain $\mathcal{A}_2$. As a special case, when $\beta=1$, i.e., $\exists$ only one gateway in $S$, then $M^{(\beta)}=\mathcal{D}$.
When $\beta>1$, the probability $\prob(M^{(\beta)}\leq d) = \prob(\min(\mathcal{D}_1,\mathcal{D}_2,\ldots, \mathcal{D}_\beta) \leq d)$, i.e., at least one of $\{\mathcal{D}_i\}_{i=1}^\beta$ is smaller than or equal to $d$. Therefore, let $F_{\mathcal{D}}(x)$ be the cdf of $\mathcal{D}$, and $F_{M^{(\beta)}}(x)$  the cdf of $M^{(\beta)}$. Then 
\vspace{-.5em}
\begin{equation}
\label{eq:cdfD1beta}
F_{M^{(\beta)}}(x) = 1-(1-F_{\mathcal{D}}(x))^{\beta}.
\vspace{-.5em}
\end{equation}
Therefore, the pdf of $M^{(\beta)}$ is 
\vspace{-.5em}
\begin{equation}
\label{eq:pmfD1beta}
  f_{M^{(\beta)}}(x)= 
  \begin{cases}
    \!\begin{aligned}
       & (1-F_{\mathcal{D}}(x-1))^{\beta}  \\
       &  - (1-F_{\mathcal{D}}(x))^{\beta}    
    \end{aligned}   & \text{for}\ x \geq 1, \\
   \\ 1-(1-F_{\mathcal{D}}(0))^{\beta}& \text{for}\ x =0.
  \end{cases}
  \vspace{-.5em}
\end{equation}
With (\ref{eq:pmfD1beta}), we derive
\vspace{-.8em}
\begin{equation}
\label{eq:firstDomainSS}
\mathbb{E}[L(\mathcal{P}'_1)]=\mathbb{E}[M^{(\beta)}]+1=\int_{x=0}^{+\infty}xf_{M^{(\beta)}}(x)+1.
\vspace{-.5em}
\end{equation}
Substituting (\ref{eq:Delta}), (\ref{eq:firstDomainSS}), and (\ref{eq:expD}) into (\ref{eq:pathLengthExpSS}), we get the expression of $L_{\SS}$. Comparing to $L_{\MS}$, the expression of $L_{\SS}$ is more complicated as the link weight can be of any distribution. Nevertheless, it is verifiable that $L_{\SS}$ is smaller than $L_{\MS}$ yet still bounded by $L_{\MS}$ (a logarithmic function). 

\section{Average Path Length under Partial Synchronization}
\label{sec:APL_PS}

For partial synchronization (PS), the synchronization cost among controllers is higher than that under SS. In this section, we quantify how such increased synchronization cost improves the APL in the network.
\vspace{-.3em}
\subsection{Partial Synchronization Model}
\vspace{-.3em}
We consider a simple PS model: In the two-layer network, given integer $\tau$ ($\tau\geq 1$), any two domains with their shortest domain-wise distance less than or equal to $\tau-1$ are synchronized. In other words, one domain is synchronized with every other domain that is within $(\tau-1)$-hop away. Moreover, we assume that the synchronization cost between two domains is proportional to their shortest domain-wise path length. Therefore, PS with a higher $\tau$ is more difficult to achieve. In particular, when $\tau=1$, PS is reduced to SS; when $\tau$ is greater than the longest domain-wise distance between any two domains, PS can augment to CS. Note that this model does not specify how the synchronization information  is used, it only provides a synchronization radius ($\tau$). In the rest of this section, we assume that $\tau$ is not significantly large, i.e., PS$\neq$CS, no single controller has the full network knowledge. \looseness=-1
\vspace{-.3em}
\subsection{Routing Mechanism under PS (RMPS)} 
\label{sec:RMPS}
\vspace{-.3em}
Under the above partial synchronization  model, we now describe a basic routing mechanism under PS, RMPS.

\begin{enumerate}[(i)]
\item  Since PS$\neq$CS, similar to RMMS, for the given two nodes $v_1$ and $v_2$, RMPS first selects the shortest domain-wise path $\mathcal{P}(v_1,v_2)$, which goes through vertices $\vartheta(\mathcal{A}_1),\vartheta(\mathcal{A}_2),\ldots,\vartheta(\mathcal{A}_q)$ in $\mathcal{G}_d$ ($v_1\in\mathcal{A}_1,v_2\in\mathcal{A}_q$), with ties (if any) broken randomly. 
\item Following the direction from $\vartheta(\mathcal{A}_1)$ to $\vartheta(\mathcal{A}_q)$, RMPS partitions $\mathcal{P}(v_1,v_2)$ into $\lceil q/\tau\rceil$ non-overlapping sub-paths $\{\Lambda_i\}$, where vertices $\vartheta(\mathcal{A}_{(i-1)\tau+1})$, $\vartheta(\mathcal{A}_{(i-1)\tau+2})$, $\ldots$, $\vartheta(\mathcal{A}_{\min(q,(i-1)\tau+\tau)})$ are in $\Lambda_i$; $\Lambda_i$ and $\Lambda_{i+1}$ are connected in the domain-wise topology. Domains $\mathcal{A}_{(i-1)\tau+1}$, $\mathcal{A}_{(i-1)\tau+2}$, $\ldots$, $\mathcal{A}_{\min(q,(i-1)\tau+\tau)}$ are called \emph{routing cluster} $i$, as illustrated in Fig.~\ref{fig:Cluster_example}.
\item Viewing each routing cluster as a whole ``domain'', path segments are constructed in each routing cluster similar to step (ii) of RMMS, except that paths are computed based on link weights. Let $\mathcal{S}'_i$ be the constructed path segment by controllers in routing cluster $i$. Note that similar to RMMS, $\mathcal{S}'_i$ is the concatenation of path $\mathcal{S}_i$ within routing cluster $i$ and a link connecting to routing cluster $i+1$ when $i<\lceil q/\tau\rceil$;
\item The final end-to-end path $\mathcal{P}$ connecting $v_1$ and $v_2$ is
\vspace{-.5em}
\begin{equation}
\label{eq:pathPS}
\mathcal{P}=\mathcal{S}'_1+\mathcal{S}'_2+\ldots +\mathcal{S}'_{\lceil q/\tau \rceil}; 
\vspace{-.5em}
\end{equation}
\vspace{-.5em}
\end{enumerate}

\begin{figure}[tb]
\centering
\includegraphics[width=3.5in]{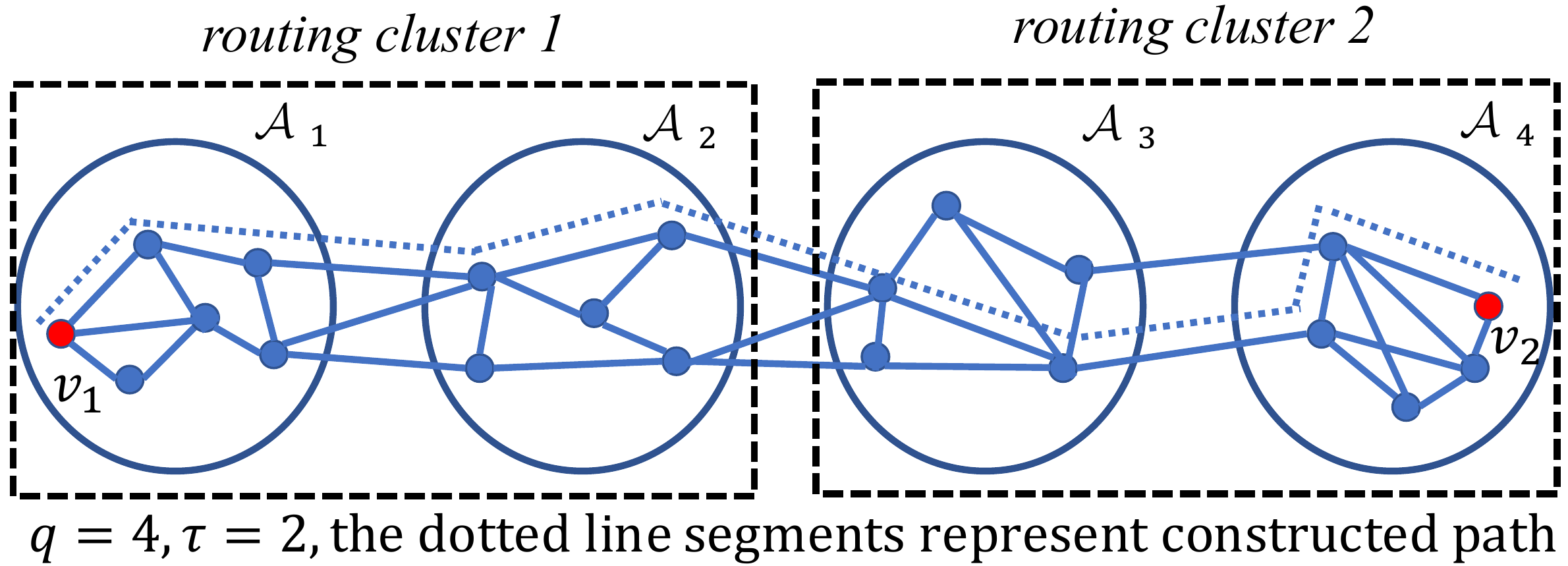}
\caption{Path construction under RMPS w.r.t. $v_1$ and $v_2$, whose shortest domain-wise path traverses $\mathcal{A}_1$, $\mathcal{A}_2$, $\mathcal{A}_3$, and $\mathcal{A}_4$. $\mathcal{A}_1, \mathcal{A}_2$ form routing cluster 1 and $\mathcal{A}_3, \mathcal{A}_4$ form routing cluster 2. Domains $\mathcal{A}_2$ and $\mathcal{A}_3$ are synchronized, the corresponding synchronized information is not used under RMPS.\looseness=-1}
\label{fig:Cluster_example}
\vspace{-.5em}
\end{figure}
\vspace{-1em}
\emph{Remark:} RMPS provides a basic path construction method for the PS scenario. Note that though domains in different routing clusters may synchronize with each other (e.g., domains $\mathcal{A}_2$ and $\mathcal{A}_3$ in Fig.~\ref{fig:Cluster_example}), the synchronized information may not used under RMPS. We acknowledge that there may exist other routing algorithms with better performance. However, our intention here is \emph{not} to design an optimal routing mechanism for PS. Instead, our goal is to quantify the performance of a given routing mechanism for PS, for which we select RMPS as a representative routing mechanism to analyze. 
\subsection{APL under Partial Synchronization (PS)}
\label{subsec:APL_PS}
To compute APL incurred by RMPS under PS, denoted by $L_{\PS}$, we first consider APL w.r.t. two nodes with their shortest domain-wise path containing exactly $q$ vertices, denoted by $L_q^{\PS}$. Similar to (\ref{eq:pathHopExpShort}), we know from (\ref{eq:pathPS}) that $\mathbb{E}[\mathcal{S}_1]=\mathbb{E}[\mathcal{S}_2]=\ldots=\mathbb{E}[\mathcal{S}_{\lceil q/\tau \rceil-1}]\neq\mathbb{E}[\mathcal{S}_{\lceil q/\tau \rceil}]$. Therefore, it suffices to determine $\mathbb{E}[\mathcal{S}_1]$ and $\mathbb{E}[\mathcal{S}_{\lceil q/\tau \rceil}]$. For ease of presentation, we introduce the  concept of bus networks.
\begin{definition}
\label{def:busnet}
A \emph{bus network} of length $k$, a special graph following the two-layer network model, consists of $k$ domains, where its domain-wise topology is a connected line graph (i.e., a special tree with each vertex having at most one child). 
\end{definition} 
In a bus network of length $k$ with the involved domains labelled as $\mathcal{A}_1$, $\mathcal{A}_2$, $\ldots$, $\mathcal{A}_k$, let random variable $D_k^{(\beta)}$ denote the length of the shortest path between two arbitrary nodes $v_1\in \mathcal{A}_1$ and $v_2\in\mathcal{A}_k$, $M_k^{(\beta)}$ the length of the shortest path between an arbitrary node $v_1\in \mathcal{A}_1$ and the closest gateway in $\mathcal{A}_k$ connecting to domains outside this bus network. Hence, to determine $\mathbb{E}[\mathcal{S}_1]$ and $\mathbb{E}[\mathcal{S}_{\lceil q/\tau \rceil}]$, we only need to compute $\mathbb{E}[M_k^{(\beta)}]$ and $\mathbb{E}[D_k^{(\beta)}]$. Thus, we define random variable $X^{(k)}:=D_{k-1}^{(\beta)}+\mathcal{D}+1$ (the pdf of $\mathcal{D}$ is in (\ref{eq:InterDomainDistance})). Let $X_1^{(k)},X_2^{(k)},\ldots,X_\beta^{(k)}$ be i.i.d. random variables following the same distribution as $X^{(k)}$. Then similar to (\ref{eq:minRV}), we have
\begin{equation}
D_{k}^{(\beta)}=\min(X_1^{(k)},X_2^{(k)},\ldots,X_\beta^{(k)}),
\vspace{-.5em}
\end{equation}
where $D_1^{(\beta)}=\mathcal{D}$. Analogously, let 
$Y^{(k)}:=D_{k-1}^{(\beta)}+M^{(\beta)}+1$, and  $Y_1^{(k)},Y_2^{(k)},\ldots,Y_\beta^{(k)}$ be i.i.d. random variables following the same distribution as $Y^{(k)}$. Then
$M_{k}^{(\beta)}=\min(Y_1^{(k)},Y_2^{(k)},\ldots,Y_\beta^{(k)})$, where $M_1^{(\beta)}=M^{(\beta)}$ is defined in (\ref{eq:minRV}).
Then following the same method in (\ref{eq:cdfD1beta}--\ref{eq:firstDomainSS}), $\mathbb{E}[D_k^{(\beta)}]$ and $\mathbb{E}[M_k^{(\beta)}]$ are computable. Note that it is expensive to compute $D_k^{(\beta)}$ and $M_k^{(\beta)}$, as they are defined in a recursive way; more efficient computation methods are discussed in Section~\ref{sec:APL_CS}.

Next, depending on the value of $q$, we have
\begin{equation}
L_q^{\PS}= 
  \begin{cases}   
   (\lfloor{q/\tau}\rfloor-1)\cdot(\mathbb{E}[M_\tau^{(\beta)}]+1)+\mathbb{E}[D_\tau^{(\beta)}]  &    \text{if }\theta=0; \\
   \\  \lfloor{q/\tau}\rfloor\cdot(\mathbb{E}[M_\tau^{(\beta)}]+1)+\mathbb{E}[D_\theta^{(\beta)}]  &    \text{if }\theta>0,
  \end{cases}
\end{equation} 
where $\theta=q\mod \tau$. Recall that two arbitrary nodes with their domain-wise path length containing $q$ vertices happen with probability being approximately $z'_{q-1}/(m-1)\approx z'_{q-1}/m$. Therefore, $L_{\PS}=\sum_{q=2}^{h'_{\max}+1}L_q^{\PS}z'_{q-1}/m$, where $h'_{\max}:=\arg\max_i z'_i$ s.t. $1+\sum_{i=1}^{h_{\max}}z'_i\leq m$.

\section{Average Path Length under Complete Synchronization}
\label{sec:APL_CS}

For complete synchronization (CS), all SDN domains are synchronized, which is equivalent to the case where there exists a logical centralized controller in the network. Therefore, the routing mechanism for CS is simple, and all controllers make the unanimous global optimal decisions. Since CS represents the best synchronization level, the following analytical results also serve as a performance bound that no routing mechanisms for any synchronization levels can exceed.

Given two arbitrary nodes $v_1$ and $v_2$, suppose the shortest domain-wise path $\mathcal{P}^*$ w.r.t. $v_1$ and $v_2$ contains $k$ vertices in the domain-wise topology.
If $\mathcal{P}^*$ corresponds to the shortest path between $v_1$ and $v_2$, then the APL under CS can be easily obtained as we have derived a method to compute $\mathbb{E}[D_k^{(\beta)}]$, the APL w.r.t. two random nodes at the end-domains in a bus network of length $k$, in Section~\ref{sec:APL_PS}. However, for the global shortest path  $\mathcal{P}^*$, it is possible that $\mathcal{P}^*$ visits more than $k$ domains while experiencing a shorter end-to-end path length. We, therefore, derive the properties of $\mathbb{E}[D_k^{(\beta)}]$ and examine how it is related to $\mathcal{P}^*$. Let $L_k(\beta):=\mathbb{E}[D_k^{(\beta)}]$. Then 
\begin{theorem}
\label{thm:shortIsGood}
For the two-layer network model, $L_{k}{(\beta)}<L_{k+1}{(\beta)}$ when $k \geq 3$. 
\end{theorem}
\begin{proof}
Consider two bus networks of length $k$ and $k+1$, whose domains are labelled as $\mathcal{A}_1,\mathcal{A}_2,\dots,\mathcal{A}_k$ and $\mathcal{B}_1,\mathcal{B}_2,\dots,\mathcal{B}_{k+1}$, respectively. Let $V_{\text{in}}(\mathcal{C}_i)$ ($V_{\text{out}}(\mathcal{C}_i)$) be the set of gateways in domain $\mathcal{C}_i$ connecting to domain $\mathcal{C}_{i-1}$ ($\mathcal{C}_{i+1}$).
Without loss of generality, we assume that $\mathcal{A}_1 = \mathcal{B}_1$ and $\mathcal{A}_k = \mathcal{A}_{k+1}$. This  implies that $V_{\text{out}}(\mathcal{A}_1)=V_{\text{out}}(\mathcal{B}_1)$ and $V_{\text{in}}(\mathcal{A}_k)=V_{\text{in}}(\mathcal{B}_{k+1})$.
Therefore, $L_{k}{(\beta)}$ and $L_{k+1}{(\beta)}$ are determined by the pair-wise distance between $V_{\text{out}}(\mathcal{A}_1)$ and $V_{\text{in}}(\mathcal{A}_k)$, and $V_{\text{out}}(\mathcal{B}_1)$ and $V_{\text{in}}(\mathcal{B}_{k+1})$, respectively. When $k \geq 3$, there exist at least one domain, called \emph{middle domain}, apart from the end-domains in a bus network. Since all middle domains have the same statistical parameters, each middle domain offers the same probability of finding a path with certain APL in that domain. Furthermore, more middle domains introduce more inter-domain edges. Thus, there is no higher possibility of finding a shorter route due to the presence of more middle domains. Therefore, more middle domains result in higher expectation of pair-wise distance between $V_{\text{out}}(\mathcal{B}_1)$ and $V_{\text{in}}(\mathcal{B}_{k+1})$. 
\end{proof}

Theorem~\ref{thm:shortIsGood} reveals a property of $L_{k}{(\beta)}$ that is vital to our analysis, i.e., a longer domain-wise path incurs higher end-to-end path weight if the shortest domain-wise path between two nodes contains at least three vertices. Note that in Theorem~\ref{thm:shortIsGood}, there are two uncovered cases. First, $k=1$. Since we are not interested in determining APL for two random nodes within the same domain, this case is ignored. Second, $k=2$. From numerical results, we observe that $L_{2}{(\beta)}$ may be slightly greater than $L_{3}{(\beta)}$ when $\beta$ satisfies certain conditions. Nevertheless, the case that two random nodes residing in two neighboring domains only happen with probability $z'_1/m$, which can be ignored as $z'_2\gg z'_1$ and $m$ is large. Next, based on Theorem~\ref{thm:shortIsGood}, the following two corollaries on $L_{k}{(\beta)}$ can be deduced, which always hold irrespective of the values of $k$.

\begin{corollary}
\label{cor:beta=1}
For the two-layer network model, $ L_{k+1}{(1)} -L_k{(1)} = \mathbb{E}[\mathcal{D}]+1$.
\end{corollary}
\begin{proof}
When $\beta = 1$, 
$L_k(1)=\mathbb{E}[D_{k}^{(1)}] = k\mathbb{E}[\mathcal{D}] + k - 1 $, and $L_{k+1}(1) = \mathbb{E}[D_{k+1}^{(1)}] = (k+1)\mathbb{E}[\mathcal{D}] + k$. Therefore, $ L_{k+1}(1) -L_k(1) = \mathbb{E}[D_{k+1}^{(1)}]- \mathbb{E}[D_{k}^{(1)}] = \mathbb{E}[\mathcal{D}]+1$. 
\end{proof}
\begin{corollary}
\label{cor:beta=infty}
For the two-layer network model, $\lim_{\beta \to \infty} \big(L_{k+1}{(\beta)} -L_k{(\beta)}\big) = 1$.
\end{corollary}
\begin{proof}
In a bus network with $k$ domains $\mathcal{A}_1,\mathcal{A}_2,\dots,\mathcal{A}_k$, when $\beta \to \infty$, every node in domain $\mathcal{A}_i$ directly connects to all nodes in domain $\mathcal{A}_{i+1}$ ($i\leq k-1$). As a result, the APL within each domain on a bus network is $0$. Thus, the APL is the sum weight of all traversed inter-domain edges, which is $k$ for $L_{k+1}(\beta)$ and $k-1$ for $L_k(\beta)$. 
\end{proof}

Note that one implicit assumption for Theorem~\ref{thm:shortIsGood} is that the domain-wise path associated with the constructed path is a \emph{simple path}, i.e., a path without repeated vertices. To show that visiting more domains cannot construct a shorter end-to-end path, we still need to prove that visiting one domain more than once is also disadvantageous. To this end, we define $L'_k(\beta)$ which is similar to $L_k(\beta)$ except that the corresponding domain-wise path contains repeated vertices.
\begin{corollary}
\label{cor:repeatedDomains}
For the two-layer network model, $L_k(\beta)<L'_{k'}(\beta)$ for $k'>k$.
\end{corollary}
\begin{figure}[tb]
\centering
\includegraphics[width=3.0in]{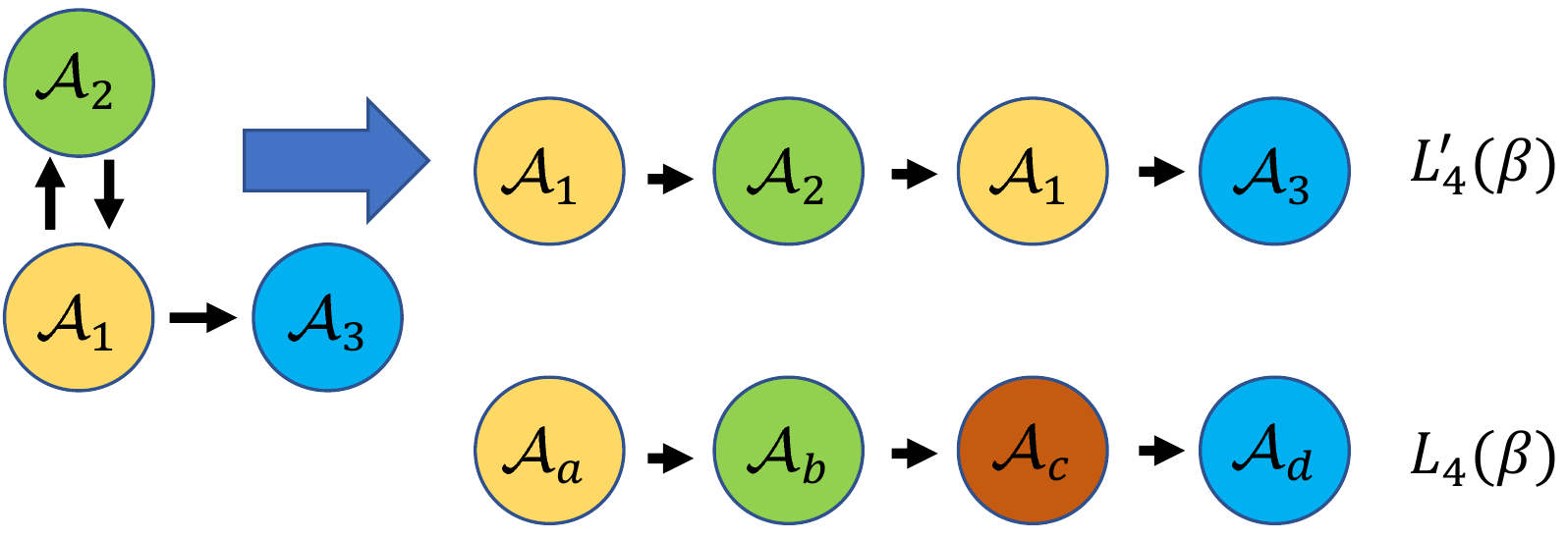}
\caption{Non-simple vs. simple domain-wise path for path constructions.}
\label{fig:proofexp}
\end{figure} 

\begin{proof}
We start the proof by comparing $L'_{k'}(\beta)$ and $L_{k'}(\beta)$. We consider the simplest form of domain repetition where only one domain is traversed twice.  We use Fig.~\ref{fig:proofexp} to facilitate the proof, where $k'=4$. Suppose that a random node in domain $\mathcal{A}_1$ needs to communicate with a random node in domain $\mathcal{A}_3$ and the selected domain-wise path is $\mathcal{A}_1-\mathcal{A}_2-\mathcal{A}_1-\mathcal{A}_3$. Apparently, this is not a simple path because domain $\mathcal{A}_1$ is traversed twice; the corresponding APL is denoted by $L'_{4}(\beta)$. 
We also consider a similar scenario with a simple domain-wise path of the same domain-wise distance, $\mathcal{A}_a-\mathcal{A}_b-\mathcal{A}_c-\mathcal{A}_d$, whose APL is denoted by $L_{4}(\beta)$. 
We observe that the computation of $L'_{4}(\beta)$ is almost the same as that in $L_{4}(\beta)$ except that there are effectively less inter-domain routing options. Hence, $L'_{k'}(\beta)>L_{k'}(\beta)$. Such analysis remains valid in cases where there are more repeated domains in the domain-wise path.
We also know from Theorem~\ref{thm:shortIsGood} that $L_{k'}{(\beta)}>L_{k'-1}{(\beta)}$. Finally, since $k'>k$, we have $L_k(\beta)<L'_{k'}(\beta)$, completing the proof.
\end{proof}

Theorem~\ref{thm:shortIsGood} together with Corollary~\ref{cor:repeatedDomains} suggest that for any source-destination node pair residing in different domains, the optimal path between them traverses the minimum number of domains with high probability. Therefore, when the shortest domain-wise path between two nodes contain $k$ vertices, then we can use $L_k(\beta)$ to approximate the corresponding optimal APL. Thus, let $L_{\CS}$ denote the APL for CS. We have
\begin{equation}
\label{eq:APL_CS}
L_{\CS} \approx \sum_{k=2}^{h'_{\max}+1}L_k{(\beta)}z'_{k-1}/m=\sum_{k=2}^{h'_{\max}+1}\mathbb{E}[D_k^{(\beta)}]z'_{k-1}/m.
\vspace{-1em}
\end{equation}
\subsection{Efficient Computation of $\mathbb{E}[D_k^{(\beta)}]$}
\label{sec:efficientEDk}
The computation of $L_{\CS}$ in (\ref{eq:APL_CS}) relies on $\mathbb{E}[D_k^{(\beta)}]$. Since $D_k^{(\beta)}$ is defined in a recursive way, it is expensive to compute the exact value of $\mathbb{E}[D_k^{(\beta)}]$. As such, we establish an efficient strategy to estimate $\mathbb{E}[D_k^{(\beta)}]$. Specifically, let $\mathcal{D}_1,\mathcal{D}_2,\ldots,\mathcal{D}_k$ denote i.i.d. random variables following the same distribution as $\mathcal{D}$. Then we define random variable $Z^{(k)}:=\sum_{i=1}^k\mathcal{D}_i+k-1$. For the two-layer network model, when the length of the bus network is increased by $1$, the number of path options w.r.t. two random nodes at the end-domains grows $\beta$-fold. Therefore, let $Z^{(k)}_1,Z^{(k)}_2,\ldots,Z^{(k)}_{\beta^{k-1}}$ be i.i.d. random variables following the same distribution as $Z^{(k)}$. Define $\widetilde{D}_k^{(\beta)}:=\min(Z^{(k)}_1,Z^{(k)}_2,\ldots,Z^{(k)}_{\beta^{k-1}})$. We then use $\mathbb{E}[\widetilde{D}_k^{(\beta)}]$ to approximate $\mathbb{E}[D_k^{(\beta)}]$. Since $\widetilde{D}_{k}^{(\beta)}$ does not rely on $\widetilde{D}_{k-1}^{(\beta)}$, $\mathbb{E}[\widetilde{D}_k^{(\beta)}]$ is easily computable using the method in (\ref{eq:cdfD1beta}--\ref{eq:firstDomainSS}). Such efficient approximation method is accurate, which is discussed in the following section.

\section{Evaluations}
\label{sec:evaluations}

To evaluate our analytical results of distributed SDN for various synchronization scenarios, we conduct a series of experiments on network topologies generated from both real and synthetic datasets. The focus in this section is two-fold. First, we validate the accuracy of the derived $L_{\MS}$, $L_{\SS}$, $L_{\PS}$, and $L_{\CS}$. We compare these theoretical results with the actual APLs collected from the above networks. Second, we aim to understand to what extent synchronization levels and network structures affect  APLs, i.e., to justify the benefit of the synchronization cost and the network structural design.
\subsection{Network Realizations}
\subsubsection{Network Topologies Based on Real Datasets}
To generate network topologies based on real datasets, we need the degree distributions as the input. Specifically, we use the real datasets collected by the University of Oregon Route Views Project (Routeview project) \cite{uoregon}, the Rocketfuel project\cite{Rocketfuel}, and the CAIDA project \cite{CAIDA} for input degree distributions.  

Given a specific degree distribution, one graph realization is generated in the following way: We assign each vertex (the total number of vertices is given) a target degree according to the degree distribution. We then select two vertices randomly and add an edge between them; the number of edges added w.r.t. each vertex is then recorded. If the degree target w.r.t. a vertex is met, this vertex will not be selected again to connect with other vertices. Such process repeats until all vertices reach their degree targets.  

\subsubsection{Network Topologies Based on Synthetic Models}
We select Barab\'{a}si-Albert \cite{barabasi1999emergence} and Erd\"{o}s-R\'{e}nyi  \cite{erdos1960evolution} models to generate network topologies.

\begin{enumerate}[(a)]
\item \emph{Barab\'{a}si-Albert (BA) model:} BA model starts with a small connected graph of a few nodes/edges. Then, we sequentially add new nodes in the following way: For each new node $v$, we connect $v$ to $\varrho$ existing nodes such that the probability of connecting to node $w$ is proportional to the degree of $w$. If the number of existing nodes is smaller than $\varrho$, then $v$ connects to all existing nodes. Vertex degree for the BA model follows a near power-law distribution. BA graphs can be used to model many naturally occurring networks, e.g., Internet and social networks.

\item \emph{Erd\"{o}s-R\'{e}nyi (ER) model:} For the ER model, the graph is generated by independently adding an edge between two nodes with a fixed probability $p$. The result
is a purely random topology where all graphs with an equal
number of links are equally likely to be selected. Vertex degree under ER follows a binomial distribution.
\end{enumerate}

Then, intra- and inter-domain topologies are generated based on the above network realization methods; see~\ref{sec:evaSet} for details. Next, on top of the generated inter-domain topologies, gateway connections are constructed according to parameter $\beta$. Moreover, we pick $p=0.015$ for ER graphs, and $\varrho = 1$ for BA graphs. In addition, all intra-domain links follow the same weight distribution, extracted from the Rocketfuel project, with the weight ranging from 1 to 16 and the expectation and variance being 3.2505 and 4.5779, respectively. 

\emph{Remark:} It should be noted that the above network realizations are only for the evaluation purpose. Our developed analytical results are generic, and are not restricted to any specific topological conditions. 

\begin{figure*}
\centering
\begin{subfigure}[b]{0.32\textwidth}
    \centering
    \includegraphics[width=\linewidth]{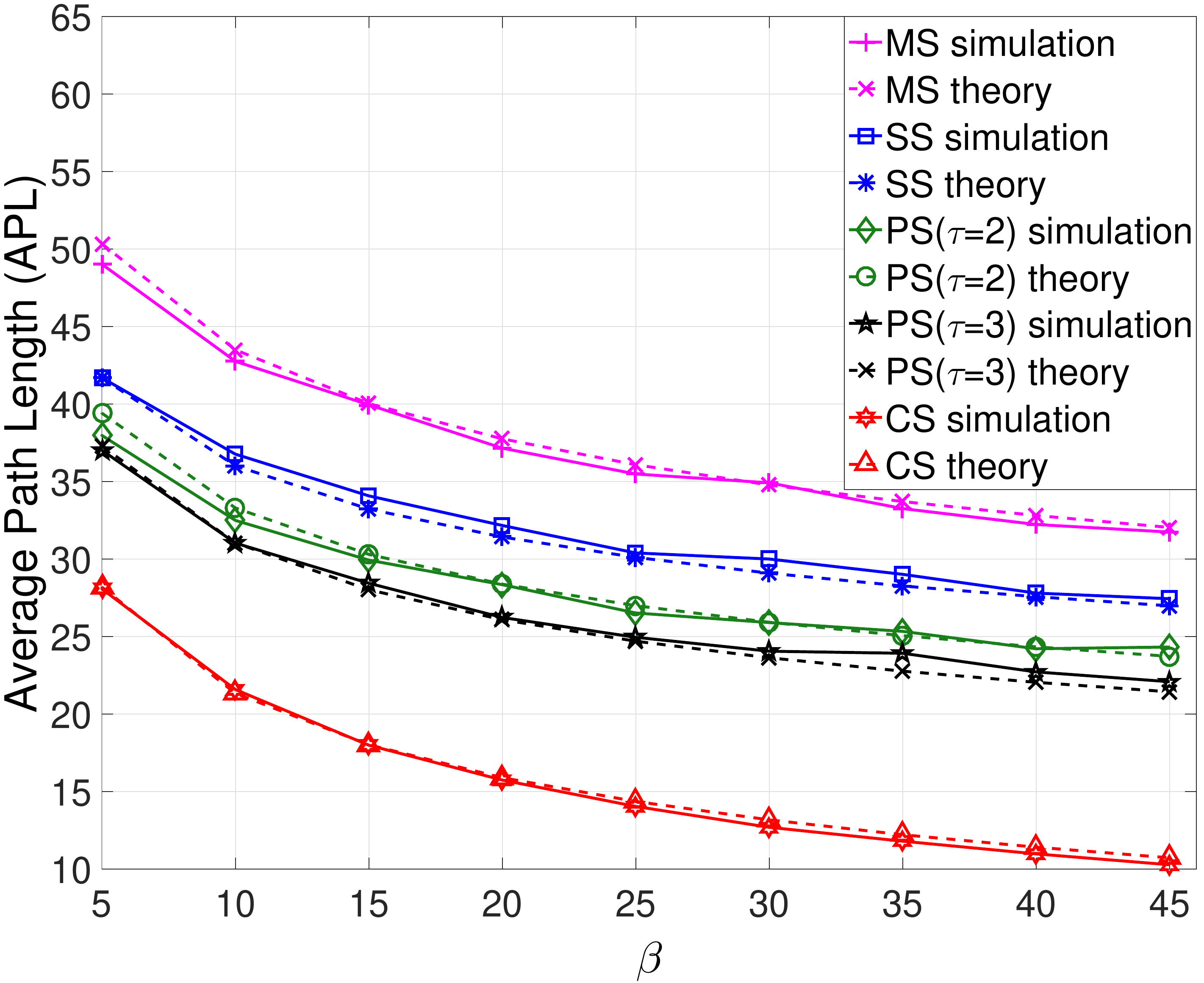}
    \vspace{-2em}
    \caption{\begin{footnotesize}Case~1 (intra-domain: Routeview-March 31, 2001; inter-domain: CAIDA-AS27524; edge weight: Rocketfuel). \end{footnotesize}}
    \label{fig:figTopology1}
\end{subfigure}\hfill
\begin{subfigure}[b]{0.32\textwidth}
    \centering
    \includegraphics[width=\linewidth]{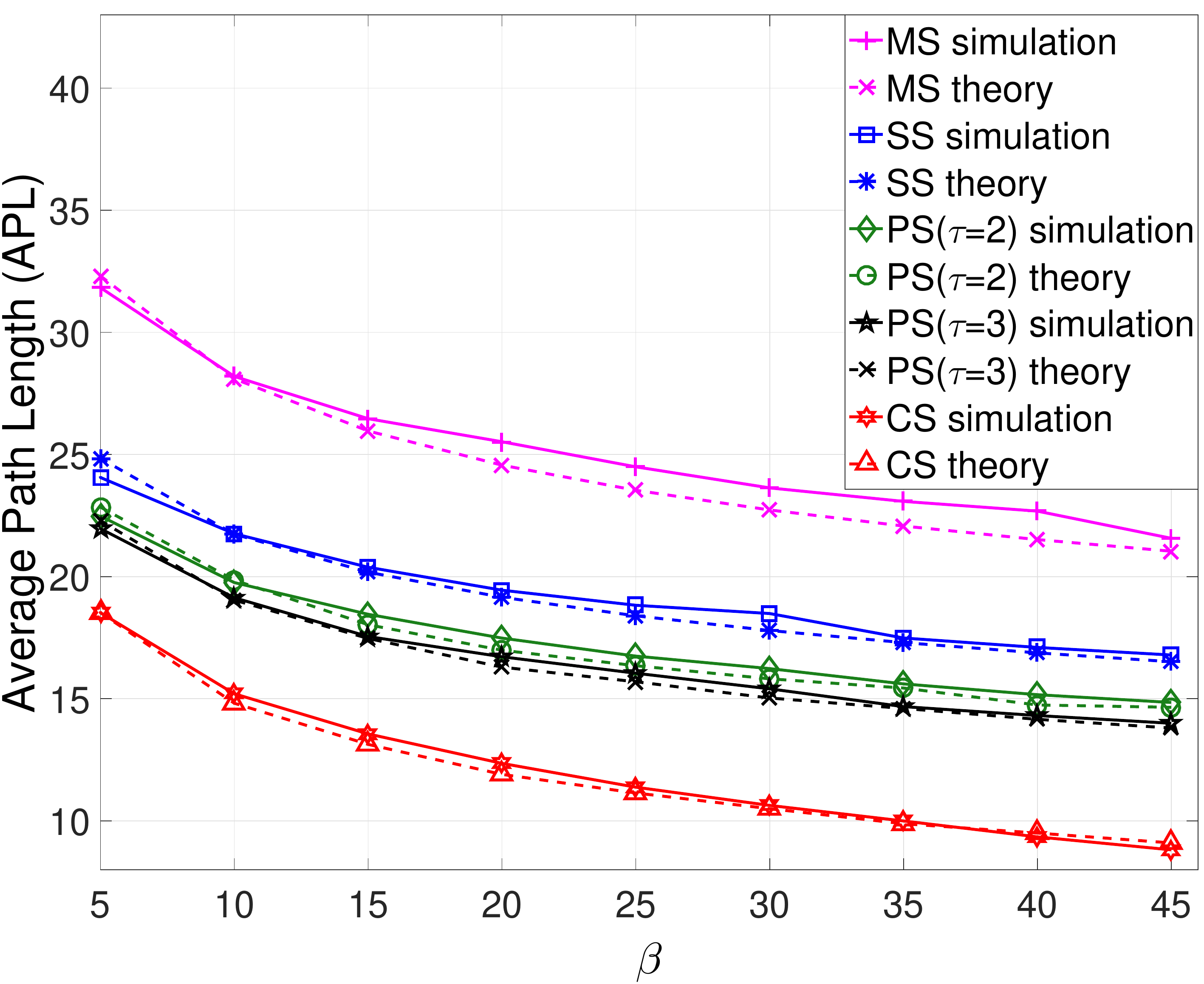}
    \vspace{-2em}
    \caption{\begin{footnotesize} Case~2 (intra-domain: Rocketfuel-AS1239; inter-domain: CAIDA-AS27524; edge weight: Rocketfuel). \end{footnotesize}}
    \label{fig:figTopology2}
\end{subfigure}\hfill
\begin{subfigure}[b]{0.32\textwidth}
    \centering
    \includegraphics[width=\linewidth]{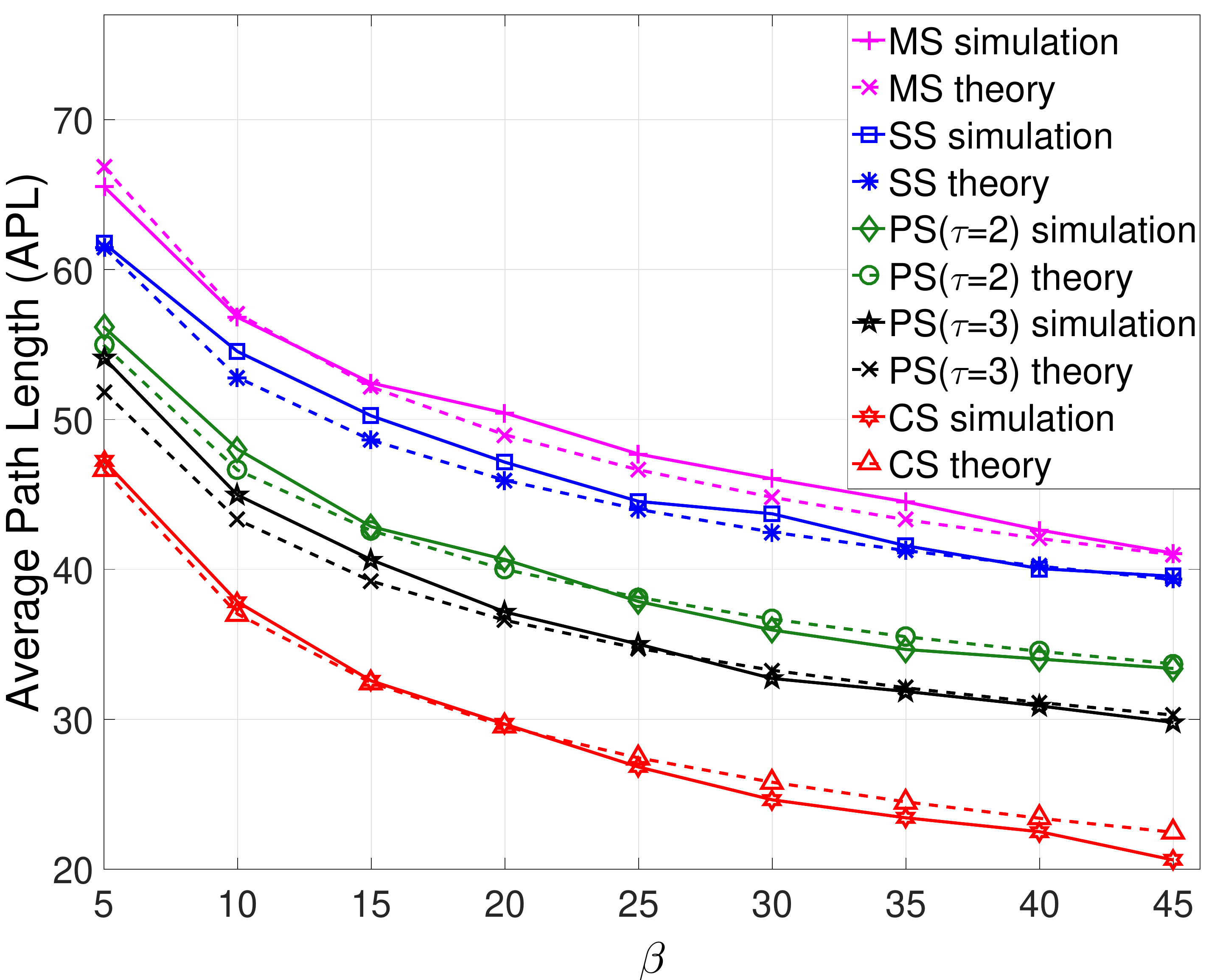}
    \vspace{-2em}
    \caption{\begin{footnotesize} Case~3 (intra-domain: Barab\'{a}si-Albert model; inter-domain: Erd\"{o}s-R\'{e}nyi model; edge weight: Rocketfuel). \end{footnotesize}}
    \label{fig:figTopology3}
\end{subfigure}
\caption{APL under different simulation cases.}
\label{fig:three-graphs}
\vspace{-2em}
\end{figure*}

\vspace{-.31em}
\subsection{Evaluation Settings}
\label{sec:evaSet}
\vspace{-.3em}
Three evaluation cases are studied: (i) Case~1, where we use degree distribution extracted from data collected on March 31, 2001 by the Routeview project to generate intra-domain topology. The inter-domain degree distribution used is based on AS27524 in CAIDA; (ii) Case~2, where the intra-domain degree distribution is calculated from AS1239 in the Rocketfuel project, and the inter-domain degree distribution is also based on AS27524 in  CAIDA; (iii) Case~3, where all intra-domain topologies are BA graphs and the inter-domain topology is an ER graph. 
For each case, the two-layer network consists of $100$ domains, each containing $200$ nodes, i.e., $m=100$ and $n=200$. For a given $\beta$, $30$ two-layer networks are realized. In each network realization, $50$ source-destination pairs (in different domains) are randomly selected to construct paths between them with MS, SS, PS, and CS.
In addition, for PS, two special cases, i.e., $\tau=2$ and $\tau=3$, are studied to compare against other synchronization scenarios. It should be noted that these settings are determined arbitrarily, as our analytical model does not require the input degree distributions to have certain patterns/properties.

\vspace{-.3em}
\subsection{Evaluation Results}
\vspace{-.3em}
The simulated APL averaged over all network realizations and source-destination node pairs are reported in Fig.~\ref{fig:figTopology1}-\ref{fig:figTopology3}, each of the evaluation cases, respectively. In these figures, each curve is also accompanied by our developed theoretical performance estimation.

\subsubsection{Accuracy of the Theoretical Results}
Evaluations of various real/synthetic networks in Fig.~\ref{fig:three-graphs}  confirm the high accuracy of our theoretical results in predicting the performance metric APL in distributed SDN networks. Specifically, the simulation curves can be closely approximated by the theoretical results for all values of $\beta$ and synchronization scenarios. Moreover, the theoretical results for PS and CS are obtained by the efficient computation method in Section~\ref{sec:efficientEDk}. Fig.~\ref{fig:three-graphs} shows that even such simplified method for approximating $L_k(\beta)$ exhibits high accuracy. 
Intuitively, this is because the process of establishing inter-domain connections in our network model is purely random, which enables us to use $\beta^{k-1}$ (see Section~\ref{sec:efficientEDk} for details) to estimate the number of route options between two random nodes in the end-domains of a bus network of length $k$.   

\subsubsection{APL Variations for Different Synchronization Levels and Structural Parameters}
Fig.~\ref{fig:three-graphs} confirms that the APL in distributed SDN is related to the amount of information available to the controllers, i.e., synchronization levels. As expected, higher synchronization levels is superior in reducing APLs. In particular, APL for CS corresponds to the minimum APL that is achievable in all cases, i.e., a lower bound. By contrast, the results for MS act as an upper bound due to the minimum intra-/inter-domain information availability. Since the APL for MS is expressed as a logarithmic function (\ref{eq:L_MS}), Fig.~\ref{fig:three-graphs} shows that even with the minimum synchronization level, APL is still significantly smaller than the network size ($20,000$ nodes in total) when all edge weights are at least $1$. 
Fig.~\ref{fig:three-graphs} shows that comparing to MS, the APL reduction for CS can be up to $60\%$. Moreover, comparing to MS, only intra-domain link weight information is available to SS. Nevertheless, such additional information is able to reduce APL by up to $50\%$. However, when more synchronized information is available, the reduction in APL starts to reduce (i.e., diminishing return). In particular, for PS, comparing against the case of $\tau=2$, the APL reduction for $\tau=3$ is rather minimal, especially when $\beta$ is small. Consequently, it is expected that with the increase of $\tau$, the benefit to cost ratio declines sharply.
Finally, as $\beta$ is a structural parameter, we observe that the network performance improves when $\beta$ increases. This is intuitive as a large $\beta$ directly renders the probability of finding a shorter route to be notably high, as there exist more inter-domain connections available for the routing mechanisms to choose from. Furthermore, Fig.~\ref{fig:three-graphs} also demonstrates that APL converges to a certain value when $\beta$ is large, which can be explained by Corollary~\ref{cor:beta=infty}. In summary, these evaluation results reveal that in distributed SDN, the performance improvement space is only marginal when the network exhibits high synchronization level and contains a large number of gateways in each domain. Such constraints need to be addressed in practical network design and optimizations.

\section{Conclusions}
\label{sec:Conclusions}
We have studied the performance of distributed SDN networks for different inter-domain synchronization levels and network structural properties from the analytical perspective. For this goal, a generic network model is first proposed to capture key attributes in distributed SDN networks. Based on this model, we have developed analytical results to quantify the performance of the constructed paths for four canonical synchronization scenarios. Extensive simulations on both real and synthetic networks show that our developed analytical results exhibit high accuracy while also providing significant insights into the relationship between the network performance and operational tradeoffs, which are vital to future network architecture and protocol design.\looseness=-1 

\bibliographystyle{IEEEtran}
\bibliography{reference}
\end{document}